\documentclass[11pt]{article}

\usepackage{amsthm}
\usepackage{amssymb}
\usepackage{amsmath}
\usepackage{amsfonts}
\usepackage{times}
\usepackage{fullpage}

\def\qed{\hfill $\vcenter{\hrule height .3mm
\hbox {\vrule width .3mm height 2.1mm \kern 2mm \vrule width .3mm
height 2.1mm} \hrule height .3mm}$ \bigskip}

\def \Sph{\mathbb{S}^{n-1}}
\def \RR {\mathbb R}

\def \EE {\mathbb E}

\def \PP {\mathbb P}
\def \eps {\varepsilon}

\def \VAR {\mathrm{Var}}

\newtheorem{theorem}{Theorem}[section]
\newtheorem{lemma}[theorem]{Lemma}

\newtheorem{proposition}[theorem]{Proposition}
\newtheorem{corollary}[theorem]{Corollary}
\theoremstyle{definition}

\theoremstyle{remark}

\long\def\symbolfootnotetext[#1]#2{\begingroup
\def\thefootnote{\fnsymbol{footnote}}\footnotetext[#1]{#2}\endgroup}

\begin{document}

\title{Efficient Algorithms for Discrepancy Minimization in Convex Sets}
\author{Ronen Eldan\thanks{Microsoft Research, Redmond. \tt{roneneldan@gmail.com}} \and Mohit Singh \thanks{Microsoft Research, Redmond. \tt{mohitsinghr@gmail.com}}}
\date{}

\maketitle

\abstract{
A result of Spencer~\cite{Spencer85} states that every collection of $n$ sets over a universe of size $n$ has a coloring of the ground set with $\{-1,+1\}$ of discrepancy $O(\sqrt{n})$. A geometric generalization of this result was given by Gluskin~\cite{Gluskin89} (see also Giannopoulos~\cite{Giannopoulos97}) who showed that every symmetric convex body $K\subseteq R^n$ with Gaussian measure at least $e^{-\epsilon n}$, for a small $\epsilon>0$, contains a point $y\in K$ where a constant fraction of coordinates of $y$ are in $\{-1,1\}$. This is often called a partial coloring result. While both these results were inherently non-algorithmic, recently Bansal~\cite{Bansal10} (see also Lovett-Meka~\cite{LovettM12}) gave a polynomial time algorithm for Spencer's setting and Rothvo\ss~\cite{Rothvoss14} gave a randomized polynomial time algorithm obtaining the same guarantee as the result of Gluskin and Giannopoulos.

This paper has several related results. First we prove another constructive version of the result of Gluskin and Giannopoulos via an optimization of a linear function. This implies a linear programming based algorithm for combinatorial discrepancy obtaining the same result as Spencer.

Our second result gives a new approach to obtains partial colorings and shows that every convex body $K\subseteq R^n$, possibly non-symmetric, with Gaussian measure at least $e^{-\epsilon n}$, for a small $\epsilon>0$, contains a point $y\in K$ where a constant fraction of coordinates of $y$ are in $\{-1,1\}$. 

Finally, we give a simple proof that shows that for any $\delta >0$ there exists a constant $c>0$ such that given a body $K$ with $\gamma_n(K)\geq \delta$, a uniformly random $x$ from $\{-1,1\}^n$ is in $cK$ with constant probability. This gives an algorithmic version of a special case of the result of Banaszczyk~\cite{Banaszczyk98}.\pagebreak

\section{Introduction}
Discrepancy problems appear in various areas of computer science and mathematics, we refer the reader to texts by Matou\v{s}ek~\cite{Matousek99} and Chazelle~\cite{Chazelle}. In the combinatorial discrepancy problem, we are given a universe $U=\{1,\ldots,n\}$ and sets $S_1,\ldots, S_m\subseteq U$ and the goal is to find a \emph{coloring} $\chi: U\rightarrow \{-1,+1\}$ that minimizes

$$\max_{j\in [m]} \left|\sum_{i\in S_j}  \chi(i)\right|.$$

A celebrated result of Spencer~\cite{Spencer85} states that there is a coloring with discrepancy $O(\sqrt{n})$ when $m=n$. There is a natural connection between discrepancy theory and convex geometry; Gluskin~\cite{Gluskin89} proved the same result as Spencer~\cite{Spencer85}, independently, using convex geometric arguments. Giannopoulos~\cite{Giannopoulos97}, building on the work of Gluskin~\cite{Gluskin89}, showed that the following generalization of Spencer's result: Given a symmetric convex body $K\subseteq \RR^n$ with Gaussian measure at least $e^{-\delta n}$, then for a small enough $\delta$, there exists $y\in K$ such that $\Omega(n)$ coordinates of $y$ are set to either $-1$ or $1$\footnote{While this result gives only a \emph{partial coloring}, by applying this result recursively, one can obtain the same result as Spencer's.}.

Interestingly, all these results were inherently non-algorithmic and obtaining polynomial time algorithms for the combinatorial discrepancy problem was highlighted as an open problem~\cite{Alon00}. Bansal~\cite{Bansal10}, in a breakthrough result, gave a polynomial time algorithm for the combinatorial discrepancy problem attaining the same discrepancy as the result of Spencer. Lovett and Meka~\cite{LovettM12} later gave a much simplified algorithm attaining the same guarantee. Both these algorithms inherently used the combinatorial structure of the problem and were not applicable to the general setting of finding a partial coloring in a convex body as given by the result of Giannopoulos~\cite{Giannopoulos97}. Recently, Rothvo\ss~\cite{Rothvoss14} gave a polynomial time algorithm that gives an algorithmic version of this result.

Another well-studied case of combinatorial discrepancy is to bound the discrepancy in terms of the maximum occurrence of any element among the $m$ sets. Beck and Fiala~\cite{BeckF81} showed that any set system has discrepancy $2t-1$ if each element appears in no more than $t$ sets and conjectured that the bound could be improved to $O(\sqrt{t})$. Techniques of Spencer~\cite{Spencer85}, and its algorithmic versions, can be adapted to bound the discrepancy by $O(\sqrt{t} \log n)$. Further improvement was obtained by Banaszczyk~\cite{Banaszczyk98} who showed a general result proving that given arbitrary unit vectors $u_1,\ldots, u_m\in \RR^n $ and convex body $K$ with $\gamma(K)\geq \frac12$ there exists signs $\epsilon_1,\ldots,\epsilon_m\in \{-1,1\}$  such that $\sum_{i} \epsilon_i u_i\in K$. This implies an improved bound of $O(\sqrt{t\log n})$ on the discrepancy of any set system where $t$ is the maximum occurrence of any element. 

\subsection{Our Results}

In this paper, we prove several algorithmic results for the discrepancy problem. A common feature of all our results is that we analyze the algorithms for the more general geometric formulations of our problem rather than the combinatorial version. These generalizations allow us to take advantage of many results in the theory of convex geometry.

Our first result shows that optimizing a random linear objective over the convex body results in a \emph{partial coloring}. Let $\gamma_n$ denote the $n$-dimensional standard Gaussian measure with density function $\frac{1}{{(2\pi)^{n/2}}}{e^{-\frac{{\|x\|^2}}{2}}}$.
\begin{theorem}\label{thm:main2}
 For any constant $0<\eps<\left(\frac{1-\sqrt{2/\pi}}{32}\right)^4$, there exists a constant $0<\delta<1$ such that every \emph{symmetric} convex body $K\subseteq \RR^n$ with $\gamma_n(K)\geq e^{-\eps n}$, the point $x=argmax\{\Gamma \cdot y: y\in K\cap [-1,1]^n\}$ where $\Gamma$ is a standard Gaussian in $\RR^n$,  satisfies  $\# \{i\in [n]:|x_i|=1 \}\geq \delta n$ with probability at least $\frac12$.
\end{theorem}

A corollary of the above result is the fact that {solving a series of linear programs} gives a coloring for the combinatorial discrepancy problem matching the result of Spencer~\cite{Spencer85}. The proof of this theorem adapts some of the ideas of Rothvo\ss~\cite{Rothvoss14} as well as the classical Uryshon's inequality.


Our next result gives a new approach that obtains a partial coloring without assuming symmetry of the convex body.
\begin{theorem}\label{thm:main1}
For any constant $\alpha\geq 0$, there exist constants $0<\eps, \delta<1$ such that every convex body $K\subseteq \RR^n$ with $\gamma_n(K)\geq e^{-\eps n}$ contains a point $x\in K$ with $\{i\in [n]:|x_i|=\alpha \}\geq \delta n$. Moreover, there is a polynomial time algorithm that given a membership oracle for $K$, returns such a point $x$ with high probability.
\end{theorem}

The algorithm uses the covariance matrix of the convex body and its restrictions. The main technical ingredient is to use the property that the measure $\gamma_K$, obtained by restricting $\gamma_n$ to the convex body $K$, is \emph{more log-concave} than the Gaussian measure.

 While Theorem~\ref{thm:main2} (and the results of Gluskin~\cite{Gluskin89}, Giannopoulos~\cite{Giannopoulos97} and Rothvo\ss~\cite{Rothvoss14}) guarantee the point $x\in K\cap [-1,1]^n$, Theorem~\ref{thm:main1} guarantees only that $x\in K$. This is necessary since the body $\{x\in \RR^n : x_1\geq 2\}$ satisfies the conditions of the theorem but does not intersect the hypercube $[-1,1]^n$. A consequence of this fact is that Theorem~\ref{thm:main1} cannot be used recursively to give an optimal coloring for the combinatorial discrepancy problem. Nonetheless, it shows that the technical condition of symmetry is not necessary if one aims to just find a partial coloring.

Our last result gives an algorithmic version of a special case of the result of Banaszczyk~\cite{Banaszczyk98} where $u_i=e_i$ for each $1\leq i\leq n$.

\begin{theorem}\label{thm:banaszcyzk}
For every $\delta>0$, there exists a constant $c\geq 0$ such the following holds. Let $K\subseteq \RR^n$ be a convex and symmetric body such that $\gamma(K)\geq \delta$ and let $x$ be a uniformly random vector from $\{-1,1\}^n$. Then
$$Pr[x\in cK]\geq \frac12.$$
\end{theorem}
%
%
%
%

 The structure of the rest of the paper is as follows. We prove Theorem~\ref{thm:main2} in Section~\ref{sec:lp}, Theorem~\ref{thm:main1} in Section~\ref{sec:main1} and Theorem~\ref{thm:banaszcyzk} in Section~\ref{sec:banas}.
 
\section{A linear programming algorithm}\label{sec:lp}

Let $K \subset \RR^n$ be a convex body and let $\Gamma = (\Gamma_1,...,\Gamma_n)$ be a standard Gaussian random vector in $\RR^n$. For $0 \neq y \in \RR^n$, set
$$
s_K(y) = \arg \max_{x \in K} \langle x, y \rangle,
$$
the supporting point of $y$ in $K$ (here, we agree that if there is more than one argument which maximizes the expression, for the purpose of analysis, we take the point closest to the origin which is unique by convexity). Note that given $K$ and $y$, the point $s_K(y)$ can be found by optimizing a linear function over $K$ which is a linear program when $K$ is a polytope.
Next, define $C = [-1, 1]^n,$ and for any $0 \neq y \in \RR^n$, we also define
$$
a(y) = \frac{1}{n}  \left [ \# \bigl \{ i; s_{K \cap C}(y)_i \in \{-1, 1 \}  \bigr \} \right ].
$$
In other words, $a(y)$ denotes the proportion of coordinates which are set to $-1$ or $+1$ in the point $s_{K \cap C}(y)$. In this notation, the proof of Theorem \ref{thm:main2} boils down to showing that for all $\eps$ small enough, there exists $\delta > 0$ such that
$$
\gamma_n(K) > e^{- \eps n} \Rightarrow \PP(a(\Gamma) \geq \delta) > c
$$
for a universal constant $c>0$.

A central definition in our proof will be the Gaussian \emph{mean-width} of a convex body, defined by
$$
w(K) := \EE[\Gamma \cdot s_K(\Gamma) ] = \EE \left [ \max_{x \in K} \langle \Gamma, x \rangle \right ].
$$

The proof, which shares some ideas with the recent proof of Rothvo\ss \cite{Rothvoss14}, relies on three classical results as its main ingredients. The first ingredient is \v{S}id\'ak's Lemma~\cite{Sidak67}:

\begin{lemma} (\v{S}id\'ak) \label{lem:sidak}
Let $K$ be a symmetric convex body and $S=\{x: |v_j\cdot x|\leq b_j\}$ be a \emph{strip}. Then $\gamma(K\cap S)\geq \gamma(K)\gamma(S)$.
\end{lemma}

The second ingredient is Sudakov-Tsirelson and Borell's well known Gaussian concentration result \cite{Borell75-2}:
\begin{theorem} \label{Gaussconc}
Let $f:\RR^n \to \RR$ be an $L$-Lipschitz function. Then one has for all $t>0$,
$$
\PP \left ( \left | f(\Gamma) - \EE[f(\Gamma)] \right  | > L t\right  ) < 2 e^{-t^2 / 2}.
$$
\end{theorem}

The last classical ingredient is known as Urysohn's inequality.

\begin{theorem} \label{Urysohn} (Urysohn's inequality)
Let $K$ be a convex body and let $B$ be a centered Euclidean ball satisfying $\gamma(K) = \gamma(B)$. Then $w(K)\geq w(B)$.
\end{theorem}
When the Gaussian measure is replaced by Lebesgue measure, this is a classic inequality in convex geometry proven in \cite{Urysohn24}. The proof for the Gaussian measure follows the same lines. For completeness, we provide a sketch of this proof.
\begin{proof} (sketch)
Let $B'$ be the centered Euclidean ball satisfying $w(B') = w(K)$. By the monotonicity of $w(\cdot)$ it is clearly enough to show that $\gamma(B') \geq \gamma(K)$. For two convex bodies $K_1, K_2$ we denote by $K_1 + K_2$ the Minkowski-sum of the two, namely
\begin{equation} \label{minkadd}
K_1 + K_1 = \{x + y; ~ x \in K_1, y \in K_2  \}.
\end{equation}
It is straightforward to check that, by definition $w(K_1 + K_2) = w(K_1)  + w(K_2)$. Let $U_1,U_2,...$ be a sequence of independent orthogonal transformations in $\RR^n$ uniformly distributed in the orthogonal group $SO(n)$. Define
$$
K_N = \frac{1}{N} \sum_{j=1}^N U_j K.
$$
Then it follows from \eqref{minkadd} and by induction that $w(K_N) = w(K)$. Moreover, since the Gaussian measure is log-concave (which follows from \cite{Borell75-1}), we have that
$$
\gamma(K_N) = \gamma \left ( \frac{1}{N} \sum_{j=1}^N U_j K \right ) \geq \left (\prod_{j=1}^N \gamma(U_j K) \right )^{1/N} = \gamma(K).
$$
Therefore, in order to prove the theorem it is enough to show that
\begin{equation} \label{limitKn}
\lim_{N \to \infty} \gamma(K_N) = \gamma(B').
\end{equation}
But remark that by definition of the body $K$ and by the strong law of large numbers we have for all $\theta \in \Sph$,
$$
\max_{x \in K_N} \langle x, \theta \rangle \to \EE \left [ \max_{x \in K_1} \langle x, \theta \rangle \right ] = \frac{w(K)}{\EE[ |\Gamma| ]}.
$$
almost surely, as $N \to \infty$. By definition of $B'$ this implies that, as $N \to \infty$,
$$
\max_{x \in K_N} \langle x, \theta \rangle \to \max_{x \in B'} \langle x, \theta \rangle, ~~ \forall \theta \in \Sph.
$$
Equation \eqref{limitKn} now follows by the continuity of the Gaussian measure of a set with respect to its support function.
\end{proof}

Urysohn's inequality gives the following simple corollary.
\begin{corollary} \label{cor:Urysohn}
Fix $\eps > 0$. Let $K \subset \RR^n$ be a convex set satisfying $\gamma_n(K) \geq e^{-\eps n}$. Then for large enough $n$, we have
$$
w(K) \geq (1-2\sqrt{\eps}) n.
$$
\end{corollary}
\begin{proof}
Denote by $B(r)$ centered Euclidian ball of radius $r$. Let $R>0$ be chosen such that $\gamma_n(B(R))=\gamma_n(K)$. An elementary calculation gives that for all $\eta > 0$,
\begin{equation}\label{ineq:ball}
\gamma_n \Bigl (B \bigl (\sqrt{n} - \eta \bigr) \Bigr ) \leq e^{-\eta^2 / 2}.
\end{equation}
Consequently, we have
$$
\gamma_n (B((1-\sqrt{2\eps})\sqrt{n})) \leq e^{-(\sqrt{2\eps})^2 n /2} = e^{-\eps n}
$$
which implies that $R\geq (1-\sqrt{2\eps}) \sqrt{n}$. Moreover Inequality~\eqref{ineq:ball} implies that

$$\EE[\|\Gamma\|]\geq (\sqrt{n}-2\sqrt{\log n})(1-e^{-2\log n})\geq \sqrt{n}-3\sqrt{\log n}$$
for large $n$ and therefore
\begin{align}
w(B(R))=\EE \left [\max_{x\in B(R)} x\cdot \Gamma \right ]\geq \EE \left [\frac{R \Gamma}{\|\Gamma\|}\cdot \Gamma \right ]= R\EE[\|\Gamma\|]\geq R\cdot  (\sqrt{n}-3\sqrt{\log n}) \geq (1-2\sqrt{\eps}) n
\end{align}
if  $\epsilon> 6\sqrt{\frac{\log n}{n}}$. An application of Theorem \ref{Urysohn} now gives
\begin{align*}
w(K) \geq w(B(R)) \geq (1-2\sqrt{\eps}) n
\end{align*}
and the corollary is proven.
\end{proof}

For $I \subset [n]$ define
$$
K(I) := K \cap \left (\bigcap_{i \in I} \{x_i \in [-1,1] \} \right ).
$$
The central Lemma needed for our proof will be the following:
\begin{lemma} \label{mainlemsec2}
Let $K$ be such that $\gamma(K) > e^{-\eps n}$. One has
\begin{equation} \label{mainlemineq}
\PP \left ( \inf_{I \subset [n] \atop |I| < \eps n}  \Gamma \cdot s_{K(I)}(\Gamma) \leq \bigl (1 - 32 \eps^{1/4} \bigr ) n \right ) \leq e^{- \eps n}.
\end{equation}
\end{lemma}
\begin{proof}
Our first step will be to show that it is legitimate to assume that $K$ is contained in a Euclidean ball of radius $2 \sqrt{n}$. Define $K' = K \cap 2 \sqrt{n} B^n$ (where $B^n$ denotes the Euclidean unit ball in $\RR^n$). The fact that $\Gamma \cdot s_{K(I)}(\Gamma) \geq \Gamma \cdot s_{K(I) \cap 2 \sqrt n B^n}(\Gamma)$ allows us to prove \eqref{mainlemineq} with $K'$ in place of $K$. Moreover, a standard calculation gives
$\gamma_n \bigl ( \RR^n \setminus 2 \sqrt{n} B^n \bigr) < e^{-n}$, so since we may assume that $\eps < \tfrac{1}{2}$, we have  $\gamma(K') \geq \frac{1}{2} e^{-\eps n}$. Therefore, from this point on we will allow ourselves assume that $K \subset 2 \sqrt{n} B^n$ by relaxing the assumption on the volume of $K$ to the assumption $\gamma(K) \geq \frac{1}{2} e^{- \eps n}$. \\

Fix $I \subset [n]$ with $|I| < \delta n$. Lemma \ref{lem:sidak} gives
\begin{equation}
\gamma_n(K(I)) \geq \gamma_n(K) \prod_{i \in I} \gamma_n \left ( \{x_i \in [-1,1] \} \right ) \geq \tfrac{1}{2} e^{-\eps n} \gamma([-1,1])^{|I|} \geq e^{- (\eps + \delta) n }.
\end{equation}
Corollary \ref{cor:Urysohn} now gives
$$
w(K(I)) \geq (1 - 2 \sqrt{\eps + \delta}) n
$$
or, in other words,
\begin{equation}
\EE \left [  \Gamma \cdot s_{K(I)}(\Gamma) \right ] \geq \bigl (1 - 2 \sqrt{\eps + \delta} \bigr ) n.
\end{equation}
Remark that, by the assumption $K \subset 10 \sqrt{n} B^n$, we have that the function
$$
y \to y \cdot S_{K(I)} (y) = \sup_{z \in K(I)} y \cdot z
$$
is $2 \sqrt{n}$-Lipschitz (here we use the fact that the supremum of $L$-Lipschitz functions is $L$-Lipschitz). Thus, by applying theorem \ref{Gaussconc} we get
$$
\PP \left ( \Gamma \cdot s_{K(I)} (\Gamma) < \bigl (1 - 2 \sqrt{\eps + \delta} - 8 \eta \bigr ) n  \right ) \leq 2 e^{- \eta^2 n }, ~~ \forall \eta > 0.
$$
By taking a union bound over all choices of $I$, we get
$$
\PP \left ( \inf_{I \subset [n] \atop |I| < \delta n}  \Gamma \cdot s_{K(I)}(\Gamma) \leq (1 - 2 \sqrt{\eps + \delta} - 8 \eta) n \right ) \leq n \left  (n \atop \lceil \delta n \rceil \right ) e^{-\eta^2 n}
$$
$$
\leq e^{ \left ( \left (1 + \log \tfrac{1}{\delta} \right ) \delta - \eta^2 \right ) n } \leq e^{ (\sqrt{\delta} - \eta^2)n }.
$$
The proof is concluded by taking $\delta = \eps$ and $\eta = 2 \eps^{1/4}$.
\end{proof}

We are finally ready to prove the main theorem of the section.

\begin{proof} [Proof of Theorem \ref{thm:main2}]
Using the fact that removing constraints which are not tight at the optimal solution does not change the optimum value, we obtain that
$$
a(\Gamma) < \delta \Rightarrow \Gamma \cdot s_{K \cap C}(\Gamma) \geq  \inf_{I \subset [n] \atop |I| < \delta n}  \Gamma \cdot s_{K(I)}(\Gamma).
$$
It follows that (choosing $\delta = \eps$)
$$
\PP \left ( \Gamma \cdot s_{K \cap C} (\Gamma) < (1 - 32 \eps^{1/4}) n \right ) < \PP(a(\Gamma) > \eps) + \PP \left ( \inf_{I \subset [n] \atop |I| < \eps n}  \Gamma \cdot s_{K(I)}(\Gamma) < (1 - 32 \eps^{1/4}) n  \right )
$$
and by Markov's inequality together with the result of Lemma \ref{mainlemsec2},
$$
\EE \left [ \Gamma \cdot s_{K \cap C} (\Gamma) \right ] \geq \left  (1 - 32 \eps^{1/4} \right )  (1 - \PP(a(\Gamma) > \eps) - e^{- \eps n} ) n
$$
(here we used the fact that $K$ contains the origin which implies that $\Gamma \cdot s_{K \cap C} (\Gamma) \geq 0$). But on the other hand
$$
\EE \left [ \Gamma \cdot s_{K \cap C} (\Gamma) \right ] \leq w(C) = \EE \left [\max_{x \in C} \langle x, \Gamma \rangle \right ]
$$
$$
= \EE \left [\sum_{i \in [n]} |\Gamma_i| \right ] = n \EE[ |\Gamma_1|] =\sqrt{\frac{2}{\pi}} n.
$$
Combining those two inequalities finally gives
$$
\PP(a(\Gamma) > \eps) > 1 - \frac{\sqrt{2}}{\sqrt{\pi} (1 - 32 \eps^{1/4})} - e^{-\eps n}.
$$
The theorem is complete.
\end{proof}

\textbf{Extension to Full Coloring}
While Theorem~\ref{thm:main1} gives only a partial coloring, it can be applied recursively to obtain the following result of Spencer\cite{Spencer85}; see Lemma 10, Rothvo\ss~\cite{Rothvoss14} for details regarding the recursion.

\begin{corollary}\label{cor:spencer}
Given a universe $U=\{1,\ldots,n\}$ and sets $S_1,\ldots, S_m\subseteq U$, there exists a coloring $\chi:U\rightarrow [-1,1]^n$ such that $\max_{i\in [m]} |\sum_{j\in S_i} \chi(j)|=O(\sqrt{n \log {2m/n}})$.
\end{corollary}

\section{A coordinate-by-coordinate algorithm for the non-symmetric case}\label{sec:main1}

%

In this section, we prove Theorem~\ref{thm:main1}. The main ingredient in the proof is Lemma~\ref{mainlem} from which the proof follows immediately. In Section~\ref{sec:algorithm}, we provide the algorithm implementing the guarantee in the lemma.
\subsection{The main lemma for the recursion}
Our goal in this section is to prove the following Lemma.
\begin{lemma} \label{mainlem}
For any constant $\alpha\geq 0$  there exist constants $0<\eta, \tau<1$ such that the following holds. Suppose that $K \subset \RR^n$ is such that $\gamma_n(K) > e^{-\eta n}$ then there exists $i \in [n]$ and $\xi \in \{-1,1\}$ such that
$$
\gamma_{n-1} \left (K \cap \{x_i = \alpha \xi \} \right ) \geq \tau \gamma_{n} (K).
$$
\end{lemma}

The proof of Theorem~\ref{thm:main1} now follows from Lemma~\ref{mainlem} by induction.  Given $\alpha \geq 0$, let $0<\eta, \tau<1$ be constants satisfying Lemma~\ref{mainlem}. Let $\eps := \frac{\eta}{4}$ and $\delta: = \frac{\eta}{2 \log \frac{1}{\tau}}$ and it is easy to check that the condition of Lemma~\ref{mainlem} continues to hold for at least $\delta n$ applications of Lemma~\ref{mainlem} giving the existence theorem. An algorithm which efficiently finds this sequence of coordinates is described in Section \ref{sec:algorithm}. \\

Before, we prove Lemma~\ref{mainlem}, we give a few definitions and preliminaries. For a subset $K \subset \RR^n$, we define $\gamma_K$ to the probability measure such that for each measurable $B\subseteq \RR^n$,
$$
\gamma_K(B) = \frac{\gamma_n(K \cap B)}{\gamma_n(K)}.
$$
We will first need the following technical estimate. Let $\|\cdot\|$ denote the Euclidian norm on $\RR^n$
\begin{lemma} \label{cortalagrand}
For all $K \subset \RR^n$ one has
\begin{equation} \label{eqcentroid}
\left \|\int x d \gamma_K(x) \right \| \leq 4 \sqrt{\log \left ( \frac{2}{\gamma_n(K)} \right )}
\end{equation}
and
\begin{equation} \label{eqtrcov}
n - 6 \sqrt{n} \sqrt{2\log \left ( \frac{4}{\gamma_n(K)} \right )} \leq \int \|x\|^2 d \gamma_K(x)
\end{equation}
\end{lemma}
\begin{proof}
We first prove \eqref{eqcentroid}. Let $X$ be a random variable distributed with law $\gamma_K$. Define $\theta = \frac{\EE X}{\|\EE X\|}$ (if the denominator is zero then \eqref{eqcentroid} follows trivially). Let $f(x)$ be the density of the variable $\langle X, \theta \rangle$. We clearly have that for each $x\in \RR$,
$$
f(x) \leq \frac{ \gamma_1(x)}{\gamma_n(K)}.
$$

 Define $g(x) = \mathbf{1}_{x \geq \alpha} \frac{ \gamma_1(x)}{\gamma_n(K)}$ where $\alpha$ is chosen such that $\int_\RR g(x) dx = 1$, i.e., $\alpha=\Phi^{-1}(\gamma_n(K))$ where $\Phi$ denotes the one dimensional Gaussian (cumulative) distribution function. Since $\int g = \int f$, we have that $g(x) \geq f(x)$ for all $x \geq \alpha$ and $g(x) \leq f(x)$ for all $x \leq \alpha$. Consequently,
$$
\int_\RR x g(x) dx - \int_\RR x f(x) dx = \int_{\RR} x (f(x) - g(x)) dx = \int_{\RR} (x - \alpha) (f(x) - g(x)) dx \geq 0.
$$
Therefore,
$$
\left \|\int x d \gamma_K(x) \right \| = \EE[\langle X, \theta \rangle ] = \int_\RR x f(x) dx \leq \int_\RR x g(x) dx
$$
$$
= \frac{\int_{ \{x \geq \alpha\} } x \gamma_1(x) dx } { \int_{ \{x \geq \alpha\} } \gamma_1(x) dx }
$$
Now, an elementary calculation gives that
$$
\left \|\int x d \gamma_K(x) \right \| \leq \frac{\int_{ \{x \geq \alpha\} } x \gamma_1(x) dx } { \int_{ \{x \geq \alpha\} } \gamma_1(x) dx } \leq  2 |\alpha| +\frac{1}{4}= 2 |\Phi^{-1} (\gamma_n(K))|+\frac{1}{4} \leq 4 \sqrt{\log \left ( \frac{2}{\gamma_n(K)} \right )}
$$
and equation \eqref{eqcentroid} is established.

We turn to the second estimate, whose proof is based on exactly the same idea only that $x \cdot \theta$ is replaced by $\|x\|$. Let $f(x)$ be the density of the variable $\|X\|$ and let $h(x)$ be the density of $\|\Gamma\|$ where $\Gamma$ is a standard Gaussian random variable in $\RR^n$. We clearly have
$$
f(x) \leq \frac{h(x)}{\gamma_n(K)}.
$$
Define $g(x) = \mathbf{1}_{|x| \leq \alpha} \frac{h(x)}{\gamma_n(K)} $ where $\alpha$ is chosen such that $\int g(x) dx = 1$. Again, since $\int g = \int f$, we have that $g(x) \geq f(x)$ for all $|x| \leq \alpha$ and $g(x) \leq f(x)$ for all $|x| \geq \alpha$, and therefore
$$
\int_\RR x^2 g(x) dx - \int_\RR x^2 f(x) dx = \int_{\RR} x^2 (g(x) - f(x)) dx = \int_{\RR} (x^2 - \alpha^2) (g(x) - f(x)) dx \leq 0
$$
or, in other words,
\begin{equation} \label{eq111}
\int_{\RR^n} \|x\|^2 d \gamma_K(x) = \int_{\RR} x^2 f(x) dx \geq \int_{\RR} x^2 g(x) dx
= \frac{\int_{ \{|x| \leq \alpha\} } x^2 h(x)dx  } { \int_{ \{|x| \leq \alpha\} } h(x)dx }.
\end{equation}
Next, we recall the following elementary fact (which follows by a straightforward calculation); if $\Gamma$ is a standard Gaussian random variable in $\RR^n$ then
$$
\PP( \vert \|\Gamma\| - \sqrt{n} \vert > t ) \leq 2 e^{-t^2/2}.
$$
It follows that
$$
\int_{|x|\leq \sqrt{n} -  \sqrt{2\log \frac{2}{\gamma_n(K)}}} h(x) dx \leq \gamma_n(K)
$$
and thus $\alpha \geq \sqrt{n} -   \sqrt{2\log \frac{2}{\gamma_n(K)}}$.
Now, we also have
$$
\int_{|x|\leq \sqrt{n} -  k\sqrt{2\log \frac{2}{\gamma_n(K)}}} h(x) dx \leq \gamma_n(K)^{k^2}.
$$
First assume that $\gamma_n(K)\leq \frac12$. We estimate

\begin{align}
\int_{ \{|x| \leq \alpha\} } x^2 h(x)dx &\geq \sum_{k=1}^{\infty}\left(\sqrt{n} -  (k+1)\sqrt{2\log \frac{2}{\gamma_n(K)} }\right)^2 \left(\gamma_n(K)^{k^2}-\gamma_n(K)^{(k+1)^2}\right)\\
&\geq \left(\sqrt{n} -  3\sqrt{2\log \frac{2}{\gamma_n(K)} }\right)^2\gamma_n(K)\label{eqn:3eqn}.
\end{align}
Thus, we have
\begin{align}
\int \|x\|^2 d \gamma_K(x) &\geq \frac{\int_{ \{|x| \leq \alpha\} } x^2 h(x)dx  } { \int_{ \{|x| \leq \alpha\} } h(x)dx}\\
&\geq \frac{\left(\sqrt{n} -  3\sqrt{2\log \frac{2}{\gamma_n(K)} }\right)^2\gamma_n(K)}{\gamma_n(K)}=\left(\sqrt{n} -  3\sqrt{2\log \frac{2}{\gamma_n(K)} }\right)^2
\end{align}
giving us the claim.

Otherwise, consider the case when $\gamma_n(K)\geq \frac12$. Using the fact that  $\int_{ \{|x| \leq \alpha\} } x^2 h(x)dx$ decreases if we decrease $\alpha$, inequality~\eqref{eqn:3eqn} implies
\begin{align*}
\int_{ \{|x| \leq \alpha\} } x^2 h(x)dx &\geq  \left(\sqrt{n} -  3\sqrt{2\log \frac{2}{\frac{1}{2}} }\right)^2\gamma_n(K)
\end{align*}
Thus we have
\begin{align}
\int \|x\|^2 d \gamma_K(x) &\geq \frac{\int_{ \{|x| \leq \alpha\} } x^2 h(x)dx  } { \int_{ \{|x| \leq \alpha\} } h(x)dx}\\
&\geq \frac{\left(\sqrt{n} -  3\sqrt{2\log 4 }\right)^2\gamma_n(K)}{\gamma_n(K)}\geq \left(\sqrt{n} -  3\sqrt{2\log \frac{4}{\gamma_n(K)} }\right)^2
\end{align}
which finishes the proof.
\end{proof}

%

An essential ingredient for the proof will be the following one-dimensional version of Caffarelli's contraction theorem \cite{Caffarelli00}.

\begin{proposition} \label{Cafarelli}
Let $\mu$ be a probability measure on $\RR$ having the form $\frac{d \mu}{dx} = f(x) = e^{-x^2/2 - V(x)}$ where $V(x)$ is a convex function. Then there exists a unique monotone, differentiable function $T: \RR \to \RR$ satisfying
\begin{equation} \label{defT}
\mu( (-\infty, T(x)] ) = \gamma((-\infty, x]), ~~ \forall x \in \RR.
\end{equation}
Moreover, the function $T$ is a contraction, namely
\begin{equation} \label{contraction1}
|T(x) - T(y)| \leq |x-y|
\end{equation}
for all $x,y \in \RR$.
\end{proposition}
The proof can be found in \cite{Caffarelli00}. For completeness, we give a heuristic proof.
\begin{proof}
By differentiating both sides of equation \eqref{defT} with respect to $x$, we see that $T$ must satisfy
\begin{equation} \label{Tderiv}
T'(x) = \frac{\gamma(x)}{f(T(x))}.
\end{equation}
Together with the boundary condition $\lim_{-\infty} T(x) = \inf supp(\mu)$ (where $supp(\mu)$ denotes the support of $\mu$), the existence of $T$ now follows from the Picard-Lidel\"of theorem (a standard ODE existence and uniqueness theorem).

Next, we want to show that $T$ is a contraction. The previous equality suggests that
$$
\log T'(x) = \frac{- x^2 + T(x)^2}{2} + V(x).
$$
By differentiating this equation twice with respect to $x$, we get
$$
(\log T'(x))'' = -1 + T'(x)^2 + T(x) T''(x) + V''(x).
$$
Now, let $x_0$ be a point where $T'(x)$ attains a local maximum, then for this point we have that the left hand side is negative and $T''(x) = 0$. Using the fact that $V'' \geq 0$, we get that
$$
T'(x_0)^2 \leq 1.
$$
Equation \eqref{Tderiv} also shows that $T'(x)$ is continuous, so it is enough to show that $T'(x)$ attains a maximum in $\RR$. This follows by approximating $\mu$ by compactly supported measures.
\end{proof}

We will need two more lemmas. The first lemma shows that the projection of a restriction of a Gaussian measure is more log-concave than the Gaussian measure. The proof follows from Prekopa-Leindler inequality.

\begin{lemma}\label{lem:projection}
Let $K$ be a convex body in $\RR^n$, $i\in [n]$ and let $f:\RR\rightarrow \RR$ be the density of the marginal of $\gamma_K$ on to the direction $e_i$, i.e., the unique (in the almost-everywhere sense) function satisfying
$$\int_B f(x)dx=\gamma_K(\{x\in \RR^n: x_i\in B\} ), \;\; \forall B\subseteq \RR \textrm{ measurable}.$$
Then the function $f$ attains the form
\begin{equation} \label{formf}
f(x)=e^{-x^2/2-V(x)}
\end{equation}
for some convex function $V(x)$. Moreover, if $X$ is a random variable with density $f(x)$ then $\VAR[X] \leq 1$.
\end{lemma}
\begin{proof}
Consider the function $g(x) = \exp \left (\frac{1}{2} \sum_{j \neq i} x_j^2 \right ) \mathbf{1}_K(x)$. Since the function $\sum_{j \neq i} x_j^2$ is convex and since $\mathbf{1}_K$ is log-concave, the function $g$ is log-concave. By the Pr\'{e}kopa-Leindler inequality, the function $h: \RR \to \RR$ defined by
$$
h(y) = \int g(x_1,..,x_{i-1}, y, x_{i+1},.., x_n) d x_1 ... d x_{i-1} dx_{i+1} ... dx_n
$$
is log-concave as well. Thus, there exists a convex function $V(x)$ such that $h(x) = exp(-V(x))$. But note that by definition of the function $f$, there exists a normalization  constant $Z>0$ such that
$$
f(x) = Z^{-1} h(x) e^{-x^2/2}.
$$
This establishes the fact that $f$ attains the form \eqref{formf}. For the second part of the lemma, we use Proposition \ref{Cafarelli} to construct a function $T$ which pushes forward the standard Gaussian measure to the measure whose density is $f$. By equations \eqref{defT} and \eqref{contraction1}, we have
$$
\VAR[X] = \int_\RR (x - \EE[X])^2 f(x) dx = \int_\RR (T(x) - \EE[X])^2 d \gamma(x)
$$
$$
\leq \int_\RR (T(x) - T(0))^2 d \gamma(x) \leq \int_\RR (x - 0)^2 d \gamma(x) = 1.
$$
The lemma is complete.
\end{proof}


\begin{lemma}\label{logconcave}

Let $V(X)$ be a convex function such that $\frac{d \mu}{dx} = f(x) = e^{-x^2/2 - V(x)}$ is a probability density. Let $\eps, \alpha>0$ be constants which satisfy
\begin{equation} \label{condaeps}
4 \eps^{2/3} < \alpha < \frac{1}{6} \sqrt{ \log \tfrac 1 \eps - \log(2 \pi)}
\end{equation}
Let $X$ be a random variable with density $f(x)$. Suppose that
\begin{equation} \label{eqexpeps}
\bigl |\EE[X] \bigr |=\left|\int x e^{-x^2/2 - V(x)} dx\right| < \eps
\end{equation}
and
\begin{equation} \label{eqeps}
Var[X]=\int \left(x - E[X]\right)^2 e^{-x^2/2 - V(x)} dx > 1 - \eps.
\end{equation}
Then we have
\begin{equation}
\max(f(\alpha), f(-\alpha)) > \frac{1}{\sqrt{2\pi}}e^{-2 \alpha^2}.
\end{equation}

\end{lemma}
\begin{proof}
Let $T(x)$ be the monotone push-forward of the standard Gaussian measure to the measure $\mu$, hence the monotone map defined by equation \eqref{defT} of Proposition \ref{Cafarelli}. According this proposition, we have that
\begin{equation} \label{contraction}
|T(x) - T(y)| \leq |x-y|, ~~ \forall x,y \in \RR.
\end{equation}
Let $u=\EE[X]$ denote the expectation of random variable $X$ with density $f(x)$. Now, by definition of $T(x)$, we have
$$
\int (T(x) - u)^2 d \gamma = \int (x - u)^2 e^{x^2/2 - V(x)} dx
$$
and by convexity together with \eqref{eqeps},
$$
\int (T(x) - T(0))^2 d \gamma \geq \int (x - u)^2 e^{x^2/2 - V(x)} dx > 1 - \eps.
$$
In other words, we have
$$
\int_{\RR} (x^2 - (T(x) - T(0))^2) d \gamma \leq \eps.
$$
Consequently,
\begin{equation}\label{eq:diff1}
\int_{\RR} |x - T(x) + T(0)| |x + T(x) - T(0)|d\gamma= \int_{\RR} (x - T(x) + T(0)) (x + T(x) - T(0)) d \gamma \leq \eps.
\end{equation}
where the first equality follows since the two terms  $(x - T(x) + T(0))$ and $(x + T(x) - T(0))$ have the same sign for all $x\in \RR$, by the fact that $T$ is a monotone contraction.

Next, we would like to show that $|T(0)|$ is bounded by a function of $\eps$. To this end, let $\delta$ be a parameter we fix later. We calculate,
\begin{eqnarray}
|T(0)|&=& \left |\int_{\RR} T(0) d\gamma \right |\leq  \left |\int_{\RR} x-T(x)+T(0) d\gamma \right | + \left |\int_{\RR} x-T(x) d\gamma \right |\label{ineq:t01}\\
&\leq & \left  |\int_{-\delta}^{\delta} x-T(x)+T(0) d\gamma \right | + \left  |\int_{x\in \RR \setminus [-\delta,\delta]} x-T(x)+T(0) d\gamma \right |+ \eps \label{ineq:t02}\\
&\leq & \int_{-\delta}^{\delta} \delta  d\gamma +  \frac{1}{\delta}\int_{x\in \RR \setminus [-\delta,\delta]}  |x-T(x)+T(0)||x + T(x) - T(0)| d\gamma+\eps\label{ineq:t03}\\
&\leq &2\delta^2 + \frac{\eps}{\delta} +\eps\label{ineq:t04}
\end{eqnarray}

where Inequality~\eqref{ineq:t01} follows from triangle inequality, Inequality~\eqref{ineq:t02} follows from the fact $\int_{\RR} xd\gamma=0$ and $|\int_{\RR} T(x)d\gamma|=|\int_{\RR} x d\mu|\leq \eps$. Inequality~\eqref{ineq:t03} follows from that $|x-T(x)+T(0)|\leq |x|\leq \delta$ for any $x\in [-\delta,\delta]$ and $|x+T(x)-T(0)|\geq |x|\geq \delta$ for any $x\in \RR\setminus [-\delta,\delta]$ using Inequality~\eqref{contraction}. Inequality~\eqref{ineq:t04} follows from standard Gaussian estimates and Inequality~\eqref{eq:diff1}.
Now choosing $\delta =\eps^{1/3}$, we obtain that $|T(0)|\leq 4\eps^{2/3}$ (note that $\eps < 1$). Condition \ref{condaeps} together with the monotonicity of $T$ finally give
\begin{equation} \label{tapositive}
T^{-1}(\alpha) \geq 0.
\end{equation}

Observe that since $T$ is a differentiable contraction, we have $T'(x)\leq 1$ for all $x$. By differentiating equation \eqref{defT} (as in equation \eqref{Tderiv}) we therefore get
\begin{equation}\label{ineq:deriv}
f(T(x))=\frac{1}{T'(x)} \frac{d}{dx} \gamma(x)\geq \frac{1}{\sqrt{2\pi}} e^{-x^2/2}, ~~ \forall x \in \RR.
\end{equation}

In light of this inequality, we learn  that it enough to show that $T^{-1}(\alpha)$ is bounded to establish a lower bound on $f(\alpha)$. Since we may replace $f(x)$ by $f(-x)$ without changing the statement of the lemma, we may assume without loss of generality that $T(0) \geq 0$. Define
$$
A = \{x>0; x - T(x) \geq \alpha+4\eps^{2/3} \} = [\beta, \infty)
$$
(if the set $A$ is the empty set, we agree that $\beta = \infty$). Now consider the case when $\beta > 2\alpha + 4\eps^{2/3}$. In this case
\begin{eqnarray*}
&& 2\alpha + 4\eps^{2/3}- T( 2\alpha + 4\eps^{2/3})\leq \alpha +4\eps^{2/3}\\
&\implies &T(2\alpha+4\eps^{2/3})\geq \alpha\\
&\implies& T^{-1}(\alpha)\leq  2\alpha+4\eps^{2/3}
\end{eqnarray*}

Together with \eqref{tapositive}, this gives $| T^{-1}(\alpha)| \leq  2\alpha+4\eps^{2/3}$. Therefore, by Inequality~\eqref{ineq:deriv} we finally get $f(\alpha)\geq \frac{1}{\sqrt{2\pi}}e^{-(3\alpha)^2/2}$.

Otherwise, we have that $\beta \leq 2\alpha + 4\eps^{2/3}$. But in this case we can write

\begin{eqnarray}
\eps &\geq& \int_{\RR} |x - T(x) +T(0)| |x + T(x) - T(0)|d\gamma\label{eqn:boundd1} \\
&\geq &\int_{x\geq \beta} |x - T(x) +T(0)| |x + T(x) - T(0)|d\gamma \\
&\geq & \int_{x\geq \beta} |\alpha+4\eps^{2/3} +T(0)| x d\gamma\label{eqn:boundd3}\\
&\geq &  \alpha\int_{x\geq \beta} xd\gamma \label{eqn:boundd4}\\
&\geq &  \alpha \frac{1}{\sqrt{2\pi}} e^{-\beta^2/2} \geq \frac{1}{\sqrt{2\pi}} e^{-(2\alpha + 4\eps^{2/3})^2/2}  \label{eqn:boundd5}
\end{eqnarray}
where Inequality~\eqref{eqn:boundd1} follows from Inequality~\eqref{eq:diff1}, Inequality~\eqref{eqn:boundd3} follows from the fact $T(x)\geq T(0)$ for each $x\geq 0$ and Inequality~\eqref{eqn:boundd4} follows from $|T(0)|\leq 4\eps^{2/3}$. Inequality~\eqref{eqn:boundd5} follows from simple estimates on Gaussian distribution. But note that the condition \eqref{condaeps} implies that $ \eps< \frac{1}{\sqrt{2\pi}} e^{-(2\alpha + 4\eps^{2/3})^2/2}$ which contradicts this inequality. The proof is complete.
\end{proof}

We are now ready to prove the main lemma
\begin{proof}[Proof of Lemma \ref{mainlem}]
Suppose that $\gamma_n(K) > e^{-\eta n}$ where $0 < \eta <1 $ is a constant determined later on, which will depend only on $\alpha$. Let $\{e_1,...,e_n\}$ be the standard basis of $\RR^n$. For all $1 \leq i \leq n$, define
$$
u_i = \left  \langle \int x d \gamma_K(x), e_i \right  \rangle
$$
and
$$
v_i =  \int \left(\langle x, e_i \rangle-u_i\right)^2 d \gamma_K(x)
$$
According to Lemma~\ref{cortalagrand}, we have
$$
\sum_{i=1}^n u_i^2 =   \sum_{i=1}^n \left(\left \langle \int x d \gamma_K(x), e_i \right  \rangle \right)^2=\left\| \int x d \gamma_K(x) \right\|_2^2 \leq  16\log\left ( \frac{2}{\gamma_n(K)} \right ) \leq 20\eta n
$$

for large enough $n$.
According to the second part of the same lemma, we have
$$
\sum_{i=1}^n v_i = \sum_{i=1}^n\int \left(\langle x, e_i \rangle-u_i\right)^2 d \gamma_K(x) = \int \|x\|^2 d \gamma_K(x) - \left\| \int x d \gamma_K(x) \right\|_2^2 \geq  n(1-10\sqrt{\eta}-20\eta)
$$

Lemma~\ref{lem:projection} implies that for each $1\leq i\leq n$,
$
 v_i \leq 1.
$
Let $I$ be uniformly chosen at random from $[n]$, then the above implies that
$$
\EE[u_I^2] \leq 20 \eta
$$
and
$$
\EE[v_I] \leq 1 - 10\sqrt{\eta}-20{\eta}
$$

Applying Markov's inequality, we have that
$$
\PP \left ( |u_I|^2 < 50\eta \mbox { and }  v_I >1-30\sqrt{\eta}-60{\eta}   \right ) > \frac14
$$
Thus there exists an $i$ which satisfies
\begin{equation} \label{goodcoordinate}
|u_i|^2 < 50\eta ~~ \mbox{ and } ~~ v_i> 1-30\sqrt{\eta}-60{\eta}.
\end{equation}
Let $f(x)$ denote the density of marginal of $\gamma_K$ on direction $e_i$. Given $\alpha$, we choose $\eps$ small enough to satisfy the condition of Lemma~\ref{logconcave}. By choosing $\eta$ to be small enough, so that equation \eqref{goodcoordinate} is satisfied (for example $\eta < \frac{\eps^2}{1000}$ suffices) we obtain that $\max\{f(\alpha), f(-\alpha)\}\geq \frac{1}{\sqrt{2\pi}}e^{-2\alpha^2}$. Thus setting $\tau =\frac{1}{\sqrt{2\pi}}e^{-(2\alpha+4\eps^{2/3})^2/2}$, we obtain that

$$\max\{\gamma_{n-1}(K\cap \{x_i=\alpha \}), \gamma_{n-1}(K\cap \{x_i=\alpha \}) \}\geq \max\{\gamma_{n}(K)f(\alpha),\gamma_n(K)f(-\alpha)\}\geq \tau \gamma_n(K)$$
as claimed.
\end{proof}

\subsection{The algorithm}\label{sec:algorithm}

In order to make the proof of Lemma \ref{mainlem} constructive, we would like to find a way of determining whether or not a coordinate $i \in [n]$ satisfies the condition \eqref{goodcoordinate}. Clearly, in order to do this, it is enough to have a good enough approximation for the covariance matrix of the Gaussian measure restricted to the body $K$. The estimation of this covariance matrix can be done using well-known sampling techniques, based on standard constructions of random walks in log-concave measures. We refer to the reader to \cite[2.2]{LV07} for the construction of two such walks, called the Ball-Walk and Hit-And-Run random walk.

Then, to get an estimate for the covariance matrix of $\gamma|_K$, we can directly apply the following result, which is an immediate consequence of Corollary 2.7 in \cite{LV07}:

\begin{theorem} (Lov\'asz-Vempala)
For any $n \in \mathbb{N}$, $\zeta > 0$ and $\delta > 0$ there exists a number $m= poly(n, 1 / \delta, \zeta)$ such that the following holds: Let $\mu$ be a log-concave probability measure whose density is $f:\RR^n \to \RR_+$ and let $v_1,...v_m$ be independent samples from the Ball-Walk of $m$ steps in $\mu$. Define for all $\theta \in \Sph$,
$$
\tilde E_\theta := \frac{1}{m} \sum_{i=1}^m \langle v_i, \theta \rangle
$$
and
$$
\tilde V_\theta := \frac{1}{m} \sum_{i=1}^m \langle v_i, \theta \rangle^2.
$$
Then with probability at least $1-\zeta$, we have for all $\theta$,
\begin{equation}
\left | \tilde E_\theta - \int \langle x, \theta \rangle d \mu(x) \right | < \delta
\end{equation}
and
\begin{equation}
\left | \tilde V_\theta - \int \langle x, \theta \rangle^2 d \mu(x) \right | < \delta.
\end{equation}
\end{theorem}

Using this theorem, within polynomial time one can have a good enough approximation for the covariance matrix of $\gamma|_K$ such that with probability at least $1-\frac{1}{n^2}$, if a coordinate $i \in [n]$ satisfies the condition \eqref{goodcoordinate} with respect to the empirical covariance matrix of the random walk, it will also satisfy the same condition for the original measure, up to a negligible error.

The above gives us an algorithm for finding a coordinate $i \in [n]$ and a sign $\xi$ which satisfy the condition of Lemma \ref{mainlem} with probability $1-1/n^2$. In order to find the partial coloring, we reiterate by considering the new convex body $K \cap \{x_i = \alpha \xi \}$. Using a union bound, this algorithm will eventually succeed with probability at least $1-1/n$.

\section{Proof of Theorem~\ref{thm:banaszcyzk}}\label{sec:banas}

For a symmetric convex body $K \subset \RR^n$ with non-empty interior, we write
$$
\Vert x \Vert_K = \inf \{\lambda > 0; ~ x \in \lambda K \}
$$
to denote the corresponding norm. Let $\Gamma$ be a standard Gaussian random vector in $\RR^n$ and let $\Psi$ be a vector distributed according to the uniform measure on $\{-1,1\}^n$. Define
$$
\Phi(K) := \EE \left [ \Vert \Psi \Vert_K^2 \right ].
$$
Our proof will rely on the Maurey-Pisier estimate \cite{MP76}, which reads
\begin{theorem} (Maurey-Pisier)
For all symmetric $K \subset \RR^n$, we have
$$
\Phi(K) \leq \frac{\pi}{2} \EE[ \Vert \Gamma \Vert_K^2 ].
$$
\end{theorem}

Note that, by definition
$$
\Phi(K) \leq \alpha^{2} \Rightarrow \PP[\Psi \in 2\alpha K]\geq \frac12.
$$
To complete the proof of Theorem~\ref{thm:banaszcyzk}, we therefore need to show that for a symmetric, convex $K \subset \RR^n$,
\begin{equation} \label{nts}
\gamma(K) \geq \delta \Rightarrow \EE \left [ \Vert \Gamma \Vert_K^2 \right  ] < C(\delta)
\end{equation}
for some constant $C(\delta)>0$ which only depends on $\delta$.

Next, note that there exists $c=c(\delta)>0$ such that $c B^n \subset K$, where $B^n$ is the unit ball. Indeed, let $c>0$ satisfy $ \frac{1}{\sqrt{2 \pi}} \int_{-c}^{c} e^{-x^2 / 2} dx = \delta$. Then for all $\theta \in \Sph$, one has $\gamma(\{x; |x \cdot \theta| < c \}) < \delta$, and therefore we must have $K \cap \{x; |x \cdot \theta| > c \} \neq \emptyset$. By the symmetry of $K$, it follows that $c B^n \subset K$.

It now follows that for all $y \in \mathbb{S}^n$ one has
$$
\Vert y \Vert_K \leq c^{-1}
$$
which implies, using the triangle inequality, that for all $x,y \in \RR^n$,
\begin{equation}
\Bigl | \Vert x \Vert_K - \Vert y \Vert_K \Bigr | \leq c^{-1} \|x-y\|_2.
\end{equation}
In other words, the function $\Vert \cdot \Vert_K$ is $c^{-1}$-Lipschitz.

For a function $f$ which is integrable with respect to $\gamma$, we define for all $0 < t < 1$
$$
P_{\gamma, f} (t) = \inf \left \{\alpha; ~ \PP(f(\Gamma) \leq \alpha) > t \right  \}.
$$
where $\Gamma$ is a standard Gaussian vector. In other words $P_{\gamma, f}(t)$ is the $t$-percentile of the variable $f(\Gamma)$.

The next theorem which is a well known estimate in Gaussian concentration, is an immediate corollary of Theorem \ref{Gaussconc}:
\begin{theorem}
For all $0 < t < 1$ there exists a constant $C=C(t)>0$ such that the following holds: let $f$ be an $L$-Lipschitz function, then for all $p \geq 1$,
\begin{equation}
\left | \EE( (f(\Gamma))^p ) - P_{\gamma,f} (t)^p \right | \leq C L^p.
\end{equation}
\end{theorem}
\noindent This theorem implies that
$$
\EE[ \Vert \Gamma \Vert_K^2 ] \leq P_{\gamma,\Vert \cdot \Vert_K}(\delta)^2 + C(\delta) c^{-2},
$$
but note that we actually have $\gamma(K) = P_{\gamma,\Vert \cdot \Vert_K} (\delta) = 1$. This implies \eqref{nts} and the proof is complete.

%

\bibliographystyle{plain}
\bibliography{const-disc}

\end{document}